\documentclass{article}
\usepackage{amsfonts}
\usepackage{amssymb}
\usepackage{amsmath,mathrsfs}
\usepackage{amsthm}
\usepackage{latexsym}
\usepackage{layout}
\usepackage{indentfirst}
\usepackage{fancyhdr}
\usepackage{epsfig}
\usepackage{graphicx}
\usepackage{pdfsync}
\usepackage{cite}

\usepackage[english]{babel}

\theoremstyle{plain}
\begingroup
\newtheorem{theorem}{Theorem}[section]

\newtheorem{proposition}[theorem]{Proposition}

\endgroup

\theoremstyle{definition}
\begingroup

\newtheorem{remark}[theorem]{Remark}

\endgroup

\theoremstyle{remark}

\mathchardef\emptyset="001F

\numberwithin{equation}{section}

\newcommand{\R}{{\mathbb R}}

\newcommand{\Z}{\mathbb Z}

\newcommand{\wtos}{\mathrel{\mathop{\rightharpoonup}\limits^*}}

\newcommand{\e}{\varepsilon}

\newcommand{\be}{\begin{equation}}
\newcommand{\ee}{\end{equation}}
\newcommand{\bes}{\begin{eqnarray}}
\newcommand{\ees}{\end{eqnarray}}

\newcommand \dps{\displaystyle }

\newcommand{\mun}{\mu_{n}}
\newcommand{\den}{\delta_{n}}

\newcommand{\lan}{\lambda_{n}}

\def\a{\alpha}

\def\Un{{\mathcal U}_{n}(I)}
\def\bUn{\overline{{\mathcal U}}_{n}(I)}

\def\sgn{{\rm sign}}
\title{Frustrated ferromagnetic spin chains: a variational approach to chirality transitions}
\author{Marco Cicalese\thanks{Zentrum Mathematik - M7, Technische Universit\"at M\"unchen, Boltzmannstrasse 3, 85748 Garching, Germany. Email: {\tt cicalese@ma.tum.de}} \and Francesco Solombrino\thanks{Zentrum Mathematik - M7, Technische Universit\"at M\"unchen, Boltzmannstrasse 3, 85748 Garching, Germany. Email: {\tt francesco.solombrino@ma.tum.de}}}

\begin{document}

\maketitle

\begin{abstract}
We study the energy per particle of a one-dimensional ferromagnetic/anti-ferromagnetic frustrated spin chain with nearest and next-to-nearest interactions close to the helimagnet/ferromagnet transition point as the number of particles diverges.  
We rigorously prove the emergence of chiral ground states and we compute, by performing the $\Gamma$-limits of proper renormalizations and scalings, the energy for a chirality transition.
\end{abstract}

\section{Introduction}

Low-dimensional magnets have attracted the attention of the scientific community in the last years (see \cite{diep} and the references therein). Among them, edge-sharing chains of cuprates provide a natural example of frustrated lattice systems, the frustration resulting from the competition between ferromagnetic (F) nearest-neighbor (NN) and antiferromagnetic (AF) next-nearest-neighbor (NNN) interactions. In this paper we study some of the multiple scale properties of these systems, focusing on a classical spin model as a first step towards the understanding of its quantum analogue (see \cite{DmiKri} for a discussion about the relation between classical and quantum models in chains of cuprates). 

We consider a minimal energy model describing the magnetic properties of one dimensional frustrated magnetic systems: the so called F-AF spin chain model. On the one-dimensional torus $[0,1]$, the state of the system is described by the values of a vectorial spin variable parameterized over the points of the lattice $\Z$. The energy of a given state of the system $u:i\in\Z\mapsto u^{i}\in S^{1}$ is 
\begin{equation}\label{intro:energy}
E(u)=-J_{1}\sum_{i\in\Z}(u^{i},u^{i+1})+J_{2}\sum_{i\in\Z}(u^{i},u^{i+2}),
\end{equation}
where $J_{1},J_{2}>0$ are the NN and the NNN interaction parameters, respectively. While the first term of the energy is ferromagnetic and favors the alignment of neighboring spins, the second, being antiferromagnetic, frustrates it as it favors antipodal next-to-nearest neighboring spins. As a result the frustration of the system depends on the relative strength of the ferromagnetic/antiferromagnetic constants.  A more refined analysis shows that the frustration can be actually measured in terms of $\alpha=J_{1}/J_{2}$. More specifically (see Proposition \ref{minimi} and Remark \ref{rem:ferro}), for $\alpha\geq 4$ the ground state of the system is ferromagnetic, while for $0<\alpha\leq4$ it is helimagnetic (see figure \ref{intro:fig-1}). The description of the ground states of the F-AF system for a choice of the parameters such that $\alpha\simeq4$ is the main aim of our analysis. In this case the system is said to be close to the ferromagnet/helimagnet transition point (examples of edge-sharing cuprates in 
the vicinity of the ferromagnetic/helimagnetic transition point can be found in \cite{Dre-etal}, while an analysis of the thermodynamic properties of such spin chains can be found in \cite{DmiKri}, \cite{Harada1988} and \cite{Harada1984}).\\

In order to study F-AF chains close to the ferromagnet/helimagnet transition point we need to perform a multiple scale analysis of the energy in \eqref{intro:energy}. We start by first scaling the functional in \eqref{intro:energy} by a small parameter $\lan$ ($\lan\to0$ as $n\to\infty$). Further setting $J_{2}=1$ (so that now the frustration parameter is now $\a=J_{1}$) and $\Z_{n}=\{i \in \Z:\  \lambda_{n} i \in [0,1]\}$ we define $E_{n}:\{u:\Z_{n}\mapsto u^{i}\in S^{1}\}\to\R$ as
\begin{equation}\label{intro:energy-n}
E_{n}(u)=-J_{1}\sum_{i\in\Z_{n}}\lan(u^{i},u^{i+1})+\sum_{i\in\Z_{n}}\lan(u^{i},u^{i+2}).
\end{equation}
It turns out that the ground states of $E_{n}$ can be completely characterized (see Proposition \ref{minimi} and Remark \ref{rem:ferro}). Neighboring spins are aligned if $J_{1}\geq 4$ (ferromagnetic order), while they form a constant angle $\varphi=\pm\arccos(J_{1}/4)$ if $0<J_{1}<4$ (helimagnetic order). In this last case the two possible choices for $\varphi$ correspond to either clockwise or counter-clockwise spin rotations, or in other words to a positive or a negative chirality (see Fig \ref{intro:fig-1}). 
\begin{figure}
\begin{center}
\includegraphics[scale=.5 ]{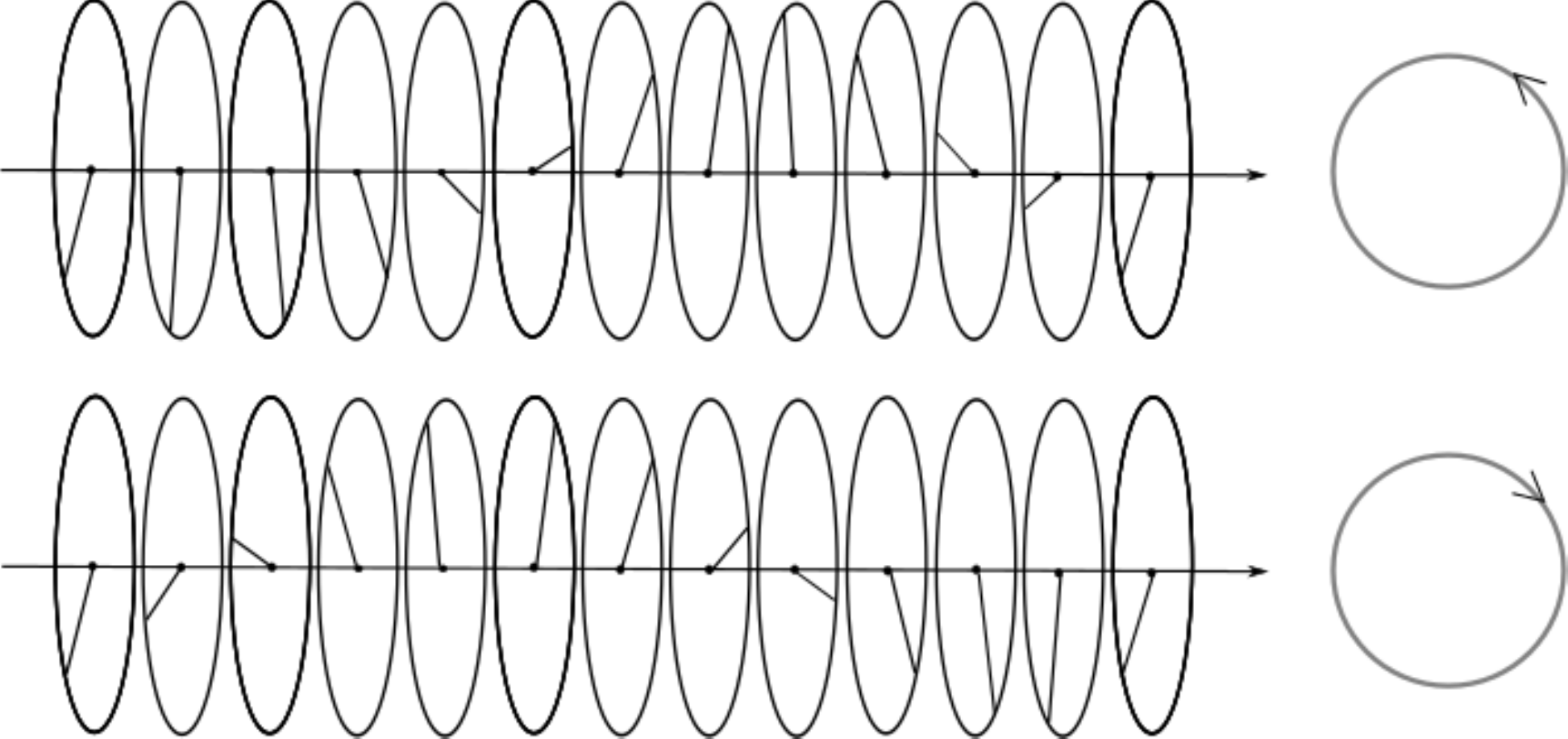}
\caption{Ground states of the spin system for $0<J_{1}<4$ for clockwise and counter-clockwise chirality}\label{intro:fig-1}
\end{center}
\end{figure}
Such a degeneracy is known in literature as {\it chirality symmetry}. The energy necessary to break this symmetry is, to the best of our knowledge, an open problem. In this paper we provide a solution to this problem in the case of a system close to the ferromagnet/helimagnet transition point, that is to say that we are able to find the correct scaling to detect the symmetry breaking and to compute the asymptotic behavior of the scaled energy describing this phenomenon as $J_{1}$ is close to $4$. Before coming to the description of our analysis, it is worth noticing that if instead of a vector spin parameter with continuous symmetry we consider a scalar one, i.e. $u\in\{\pm 1\}$, then the helicity symmetry translates into the periodicity of the ground states. In this case in \cite{BC} it has been proved that the asymptotic analysis of these systems can be performed without any restriction on the values of $J_{1}$.\\

To set up our problem we let the ferromagnetic interaction parameter $J_{1}$ depend on $n$ and be close to $4$ from below, that is in \eqref{intro:energy-n} we substitute $J_{1}$ by $J_{1,n}=4(1-\den)$ for some vanishing sequence $\den>0$. For such energies in Theorem \ref{32-th:homog-Li} we prove that, as a consequence of an abstract result proven in \cite{ACG}, their $\Gamma$-limit (with respect to the weak-$\star$ convergence in $L^{\infty}$) as $n\to\infty$ is a constant functional whose value can be approached by weakly vanishing sequences $u_{n}$ that may mix on a mesoscopic scale configurations having opposite chirality. Such a poor description suggests that, in order to get further informations on the ground states of the system we need to consider higher order $\Gamma$-limits (see \cite{GCB} and \cite{Bra-Tru} for more details as well as for the general theory of development by $\Gamma$-convergence). Note that the choice of the right energy scaling which may capture the phenomena we are interested 
in is not straightforward. In fact the continuous symmetry of the order parameter $u\in S^{1}$ adds a new difficulty: it allows for very slow variations in the angle between neighboring spins which results in the emergence of very low energy phase changes. This implies that, even if we expect to find a limit energy accounting for $0$-dimensional discontinuities of some parameter related to the chirality, the continuous symmetry of the spins makes the correct scaling not a 'surface'-type scaling. Note that this would not be the case if the spin field $u\in\{\pm 1\}$ as it is shown in \cite{ABC} (see also \cite{BC}) where the degeneracy of the ground states is solved by a surface scaling. Similar problems regarding the continuos symmetry of the order parameter arise already in \cite{ACXY} for NN systems in the context of XY spin models (see also \cite{ACP}, \cite{ADLGP}, \cite{AP} and \cite{BraCicSol} for related Ginzburg-Landau-type models). In \cite[example $1$]{ACXY} it is explicitly proved that the system 
does not undergo any phase separation that may be detected by a surface scaling. Such an example can be straightforwardly exported in the context of frustrated spin chains and, as a consequence, we are led to renormalize the energy of the system and study the asymptotic behavior of a new functional $H_{n}$ defined as
\begin{equation}\label{intro:Hn}
H_{n}(u)=\frac{E_{n}(u)-\min E_{n}}{\mu_{n}\lan}
\end{equation}
for some $\mu_{n}\to 0$ to be found. In terms of $H_{n}$, finding the energy the system spends in a transition between two states with different chirality translates into the following problem: depending on the scale $\den$
\begin{itemize}
\item[(i)] find a scaling $\mu_{n}$ and an order parameter $z_{n}$ such that if $\sup_{n}H_{n}(u_{n})\leq C$ then, as $n\to\infty$, $z_{n}$ converges to some $z$ describing a system whose chirality may have at most a finite number of discontinuities,
\item[(ii)] for such a choice of $\mu_{n}$ compute the $\Gamma$-limit of $H_{n}$ (with respect to the convergence $z_n\to z$ in the previous step) and interpret the limit functional as the energy the system spends on the scale $\lan\mu_{n}$ for a finite number of chirality transitions. 
\end{itemize}
The main result of this paper is contained in Theorem \ref{maintheorem} which states that the right scale to consider in order to keep track of energy concentration is $\lan\den^{3/2}$ (corresponding to the choice $\mu_{n}=\den^{3/2}$). We prove that, within this scaling, several regimes are possible. Roughly speaking, for $n$ large enough, we show that the spin system has a chirality transition on a scale of order $\lan/\sqrt{\den}$. As a result, depending on the value of $\lim_{n}\lan/\sqrt{\den}:=l\in[0,+\infty]$ different scenarios are possible (see Fig \ref{transizione} for a schematic picture of the transition). If $l=+\infty$ chirality transitions are forbidden (equivalently we find that the energy for a transition is infinite). If $l>0$ the spin system may have diffuse and regular macroscopic (on an order one scale) chirality transitions whose limit energy is finite on $W^{1,2}(I)$ (provided some boundary conditions are taken into account). When $l=0$ transitions on a mesoscopic scale are allowed. In 
this case the continuum limit energy is finite on $BV(I)$ and counts the number of jumps of the chirality of the system.\\

We think it is worth noticing that, to the best of our knowledge, this paper shows for the first time the presence of multiple scale regimes in a chirality transition. It is our opinion that this phenomenon is quite general and suggests that the analysis of frustrated discrete systems should take advantage from a rigorous variational method any time the parameters describing frustration and scaling may compete. As a final technical remark we would like to point out that, although our analysis is presently confined to the $1$-dimensional case, it can be easily extended to an $n$-dimensional systems for which NNN interactions are present only along the coordinate directions. Indeed, in such a case the $\Gamma$-limit of the energy of such systems can be straightforwardly obtained by a slicing procedure starting from our $1$-dimensional result.

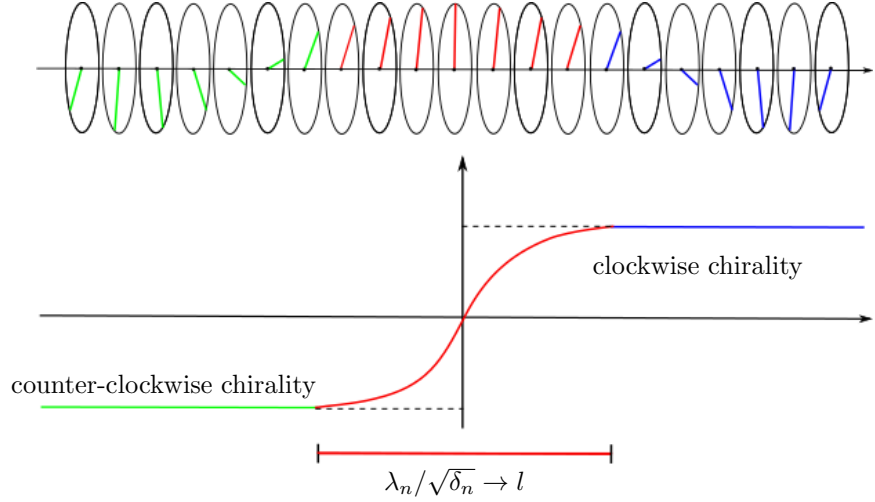
\begin{figure}
\begin{center}
\includegraphics[scale=.65 ]{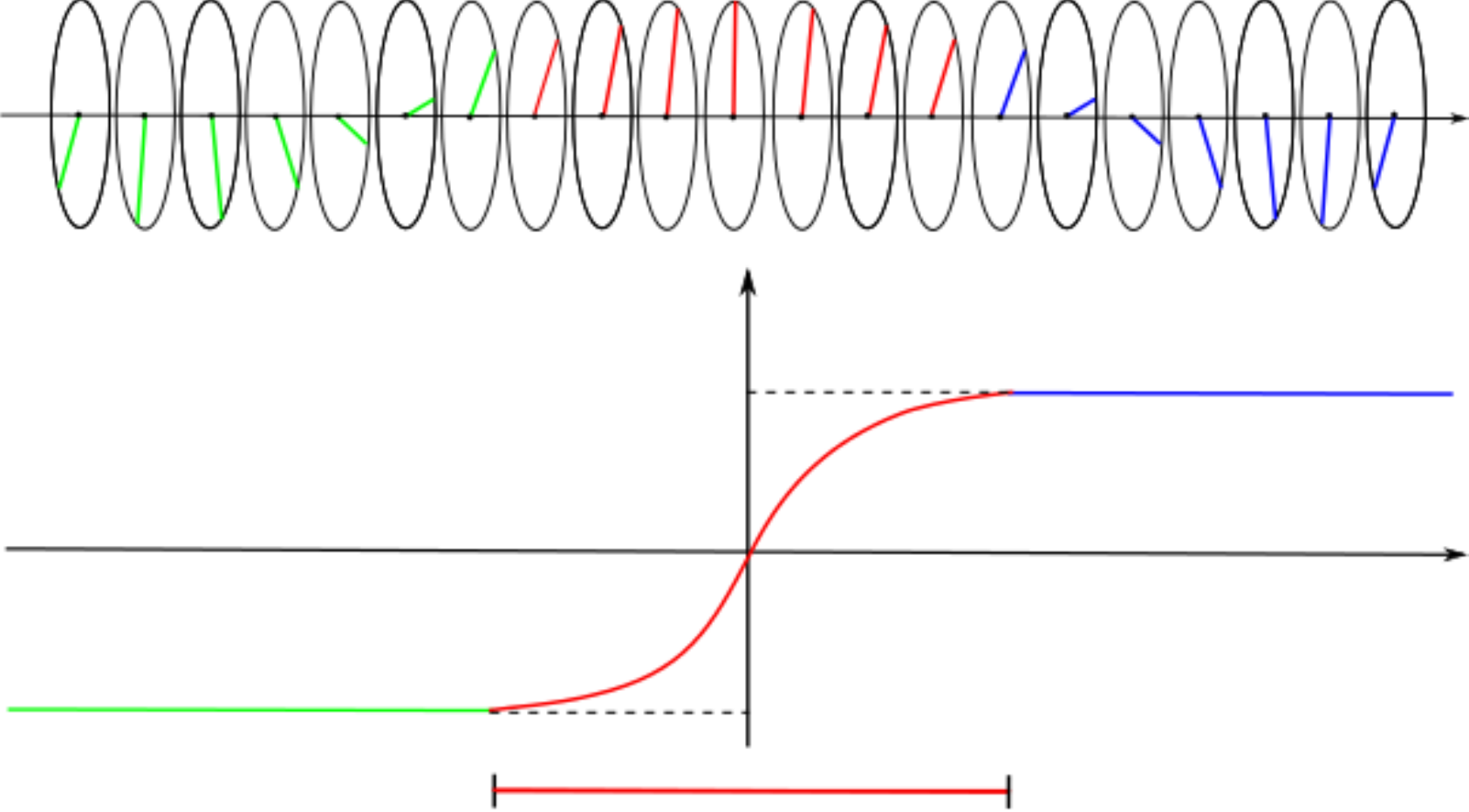}
\caption{The optimal configuration of the spins of the F-AF chain in a chirality transition (above). The transition in terms of a scalar parameter related to the chirality of the system (see Section for more details)}\label{transizione}
\end{center}
\begin{picture}(0,0)
\put(280,140){clockwise chirality}\put(200,58){$\lan/\sqrt{\den}\to l$}
\put(60,95){counter-clockwise chirality}
\end{picture}
\end{figure}
\section{Notation and Preliminaries}

Denoted by $J\subset\R$ an open interval and by $\lan$ a vanishing sequence of positive numbers, we define $\Z_n(J)$ as the set of those points $i \in \Z$ such that $\lambda_{n} i \in {\overline J}$. Given $x\in \R$, we denote by $[x]$ the integer part of $x$. The symbol $S^{1}$ stands as usual for the unit sphere of $\R^2$. Given two vectors $a,b\in\R^{2}$ we will denote by $(a,b)$ their scalar product. We will denote by ${\mathcal U}_{n}(J)$ the space of functions $u:i\in\Z_{n}(J)\mapsto u^{i}\in S^{1}$ and by $\overline{{\mathcal U}}_{n}(J)$ the subspace of those $u$ such that
\begin{equation}\label{prelim:boundary}
(u^{i_{min}+1},u^{i_{min}})=(u^{i_{max}},u^{i_{max}-1}),
\end{equation}
where $i_{min}$ and $i_{max}$ are the minimum and the maximum of $\Z_{n}(J)$, respectively. We analogously denote by ${\overline{\mathcal U}}(J)$ the space of functions $u:\Z\to S^{1}$ such that \eqref{prelim:boundary} holds with $i_{min}$ and $i_{max}$ the minimum and the maximum of $\Z\cap \overline J$, respectively.
Given $K\in\R^{m}$ we denote by $co(K)$ the convex hull of $K$. We set $Q_{h}=\left(0,h\right)$.
Given $v=(v_{1},v_{2}),w=(w_{1},w_{2})\in S^{1}$ we define the function $\chi[v,w]:S^{1}\times S^{1}\to {\pm 1}$ as  
$$
\chi[v,w]=\sgn(v_{1}w_{2}-v_{2}w_{1}),
$$ 
with the convention that $\sgn(0)=-1$.\\

We recall some preliminary results concerning the general theory of spin-type discrete systems in the bulk scaling. The following theorem has been proved in \cite{ACG}. We state it here in a version which best fits our setting.  
Let $K\subset\R^{m}$ be a bounded set. For all $\xi\in\Z$ let $f^\xi:\R^{2m}\to\R$ be a function such that
\begin{itemize} 
\item[(H1)] $f^\xi (u,v)=f^{-\xi}(v,u)$, 
\item[(H2)] for all $\xi$, $\dps{f^\xi(u,v)=+\infty}$ if $(u,v)\not\in K^2$,
\item[(H3)] for all $\xi$, there exists $C^\xi\geq 0$ such that $\dps{|f^\xi(u,v)|\leq C^\xi}$ for all $(u,v)\in K^2$, and $\sum_\xi C^\xi <\infty$.
\end{itemize}
Let us define the set of functions
\begin{eqnarray*}
D_{n}(J,\R^m)= \{u:\R\to\R^m\,:\,u \ {\rm constant \ on }\
\lan(i+[0,1))\ {\rm for \ any }\ i\in\Z_n(J)\}
\end{eqnarray*}
and the family of functionals $F_n:L^\infty(J,\R^m)\to(-\infty,+\infty]$ 
\begin{eqnarray}\label{eq:energie-eps}
F_n(u)=\begin{cases}\dps\sum\limits_{\xi\in \Z }\sum\limits_{ i
\in R_n^\xi(J)}\lan f^\xi
(u^{i},u^{i+\xi})
 & \text{if $u\in D_{n}(J,\R^m)$}\\
+\infty & \text{otherwise},\cr\end{cases}
\end{eqnarray}
where $R_{n}^\xi(J):=\{i\in\Z_n(J):\ i+\xi\in\Z_n(J)\}$. Given $v:\Z\to\R^{m}$ and $A\subset \R$ open and bounded, we define the discrete average of $v$ in $A$ as  
\begin{equation*}
\langle v \rangle_{1,A}=\frac{1}{\#(\Z\cap A)}\dps{\sum_{i  \in \mathbb{Z}\cap A}v^{i}}.
\end{equation*}

\begin{theorem} \label{32-th:homog-Li}
Let $\{f^\xi \}_{\xi}$ satisfy hypotheses (H1)-(H3). Then $F_n$ $\Gamma(w*-L^{\infty})$-converges to
$$F(u)=\int_{J} f_{hom}(u(x))dx $$
for all $u\in L^{\infty}(J,co(K))$, where $f_{hom}$ is
given by the following homogenization formula
\begin{equation}\label{32-eq:homog-form_li}
f_{hom}(z)=\lim_{\rho\to 0}\lim_{h\to + \infty} \frac{1}{h} \inf
\left\{ \sum_{\xi \in \Z} \sum_{\beta \in R^\xi_1(Q_h)}
f^\xi(v(\beta),v(\beta+\xi)),\langle v \rangle_{1,Q_h}\in
{\overline B(z,\rho)} \right\},
\end{equation}
where $R_{1}^\xi(J):=\{i\in\Z\cap J:\ i+\xi\in\Z\cap J\}$.
\end{theorem}

We now state (with minor variations) a result proved in \cite{BraYip} regarding the discrete approximations of Modica-Mortola type energies. We say that a function $W:\R\to[0,+\infty)$ is a double-well potential if it is locally Lipschitz and satisfies the following properties:
\begin{itemize}
\item[(1)] $W(z)=0$ if and only if $z\in\{\pm1\}$,
\item[(2)] $\lim_{{z\to\pm\infty}}W(z)=+\infty$,
\item[(3)] there exists $C_{0}>0$ such that $\{z:\, W(z)\leq C_{0}\}=I_{1}\cup I_{2}$ with $I_{1},I_{2}$ intervals on which $W$ is convex.
\end{itemize}
Let $\a_{n},\beta_{n}$ be two sequences of positive numbers such that $\lim_{n}\a_{n}=0$, $\lim_{n}\a_{n}/\beta_{n}=1$ and $\lim_{n}\lan/\a_{n}=0$ and let $G_{n}:L^{1}(J)\to[0,+\infty]$ be defined as
\begin{equation}\label{Gn}
G_{n}(u)=\begin{cases}\dps\a_{n}\sum_{i}\lan\left(\frac{u^{i+1}-u^{i}}{\lan}\right)^{2}+\frac{1}{\beta_{n}}\sum_{i}\lan W(u^{i})
 & \text{if $u\in D_{n}(J,\R)$}\\
+\infty & \text{otherwise},\end{cases}
\end{equation}
with $W$ a double-well potential. The following $\Gamma$-convergence result holds. 
\begin{theorem}\label{braides-yip}
Let $G_{n}:L^{1}(J)\to[0,+\infty]$ be as in \eqref{Gn}, then, with respect to the $L^{1}(J)$ convergence,
\begin{equation}\label{Glim-BraYip}
\Gamma\hbox{-}\lim_{n}G_{n}(u)=\begin{cases}
C_{W}\#(S(u)\cap J) & \text{if $u\in BV(J,\{\pm1\})$}\\
+\infty & \text{otherwise in $L^{1}(J)$},\end{cases}
\end{equation}
where $C_{W}:=2\int_{-1}^{+1}\sqrt{W(s)}\ ds$.
\end{theorem}
\begin{proof}
The proof follows by Theorem $2.1$ in \cite{BraYip} once we observe that for all $u_{n}$ such that $\sup_{n}G_{n}(u_{n})\leq C$ we have 
\begin{equation*}
\left|\frac{1}{\beta_{n}}-\frac1{\alpha_{n}}\right|\sum_{i}\lan W(u^{i})\leq C\ \frac{|\beta_{n}-\alpha_{n}|}{\alpha_{n}},
\end{equation*}
so that 
\begin{equation*}
G_{n}(u_{n})=\dps\a_{n}\sum_{i}\lan\left(\frac{u^{i+1}-u^{i}}{\lan}\right)^{2}+\frac{1}{\a_{n}}\sum_{i}\lan W(u^{i})+o(1).
\end{equation*}
\end{proof}
\begin{remark}\label{rem:W}
In the explicit case $W(s)=(1-s^{2})^{2}$ the constant $c_{W}=\frac{8}{3}$.
\end{remark}

\section{The energy model: the bulk scaling}\label{section:abstract}

In this section we introduce the F-AF model of a frustrated ferromagnetic spin chain and prove a first result concerning the $\Gamma$-limit of its bulk scaling.\\

Let $I=(0,1)$ and let us consider a pairwise-interacting discrete system on the lattice $\Z_n(I)$ whose state variable is denoted by $u:\Z_n(I)\to S^{1}$. Such a system is driven by an energy $E_n:\Un\to(-\infty,+\infty)$ given by
\begin{eqnarray*}\label{energy_{1}}
E_n(u)=-J_{1}\sum_{i=0}^{\left[1/\lan\right]-2}\lan(u^{i},u^{i+1})+J_{2}\sum_{i=0}^{{\left[1/\lan\right]-2}}\lan(u^{i},u^{i+2}),
\end{eqnarray*}
for some non negative constants $J_{1},J_{2}$. Without loss of generality we will set $J_{2}=1$, thus considering the family of energies
\begin{eqnarray}\label{energy_{2}}
E_n(u)=-J_{1}\sum_{i=0}^{\left[1/\lan\right]-2}\lan(u^{i},u^{i+1})+\sum_{i=0}^{{\left[1/\lan\right]-2}}\lan(u^{i},u^{i+2}).
\end{eqnarray}
Moreover we will consider the case $J_{1}\in(0,4]$, the case $J_{1}>4$ will be shortly discussed in Remark \ref{rem:ferro}.\\
Since we are not interested to the possible formation of boundary layers, we fix periodic boundary conditions on the system:
\begin{equation}\label{boundary}
(u^{1},u^{0})=(u^{[1/\lan]},u^{[1/\lan]-1})
\end{equation}
or equivalently $u\in\bUn$.
\begin{remark}
The periodic boundary conditions in \eqref{boundary} are an alternative to the computation of the $\Gamma$-limit of $E_{n}$ with respect to a local convergence. 
\end{remark}
As usual in the analysis of discrete systems we may embed the family of functionals on a common functional space, extending $E_{n}$ to some Lebesgue space. To this end we associate to any $u\in\bUn$ a piecewise-constant interpolation belonging to the class
\begin{equation}\label{cepsB}
C_n(I,S^{1}):=\{u\in\bUn: u(x)=u(\e i)\,\,\hbox{ if } x \in \lan(i+[0,1)),\,i \in \Z_n(I)\}. 
\end{equation}
As a consequence we may see the family of energies $E_n$ as defined on a subset
of $L^\infty(I,S^{1})$ and consider their extension on $L^\infty(I,S^{1})$. With an abuse of notation we do not relabel these functionals and set $E_n:L^\infty(I,S^{1})\to(-\infty,+\infty]$ as

\begin{equation}\label{energy_{3}}
E_{n}(u)=\begin{cases}
\dps -J_{1}\sum_{i=0}^{\left[1/\lan\right]-2}\lan(u^{i},u^{i+1})+\sum_{i=0}^{{\left[1/\lan\right]-2}}\lan(u^{i},u^{i+2})&\textit{if }u\in C_n(I,S^{1})\\
+\infty&\textit{otherwise.}
\end{cases}
\end{equation}
We now define the functional $H_{n}:L^\infty(I,S^{1})\to(-\infty,+\infty]$ as
\begin{equation}\label{Hn}
H_{n}(u)=\begin{cases}
\dps\frac{1}{2}\sum_{i=0}^{{\left[1/\lan\right]-2}}\lan\left|u^{i+2}-\frac{J_{1}}{2}u^{i+1}+u^{i}\right|^{2}&\textit{if }u\in C_n(I,S^{1})\\
+\infty&\textit{otherwise.}
\end{cases}
\end{equation}
Since $|u^{i}|=1$ for all $i\in\Z_{n}$, thanks to \eqref{boundary}, the energy in \eqref{energy_{3}} can be rewritten, in terms of $H_{n}$ as 
\begin{eqnarray}\label{enhn}
E_n(u)=H_{n}(u)-\left(1+\frac{J_{1}^{2}}{8}\right)(1-c_{n}\lan),
\end{eqnarray}
for $c_{n}=\frac{1}{\lan}-\left[\frac{1}{\lan}\right]+1\in[1,2)$, so that 
\begin{equation}\label{cn}
\sum_{i=0}^{{[1/\lan]-2}}\lan=(1-c_{n}\lan).
\end{equation}
Equality \eqref{enhn} suggests that in order to study the asymptotic properties of $E_{n}$ we can equivalently study the non negative functional $H_{n}$.
\begin{subsection}{Ground states of $H_{n}$}
In this section we characterize the global minimizers of $E_{n}$, we give upper and lower bounds on its $\Gamma$-limit as $n\to\infty$ for different values of $J_{1}$. As a corollary we show that in the case $J_{1}=4$, the continuum limit is indeed trivial.
\begin{proposition}\label{minimi}
Let $E_{n}:L^\infty(I,S^{1})\to(-\infty,+\infty]$ be the functional in \eqref{energy_{1}} and $0\le J_1 \le 4$. Then 
\begin{equation}\label{minimo}
\min_{u\in L^\infty(I,S^{1})} E_{n}(u)=-\left(1+\frac{J_{1}^{2}}{8}\right)(1-c_{n}\lan).
\end{equation}
Furthermore, a minimizer $u_n$ of $E_n$ over $L^\infty(I,S^{1})$ satisfies 
\begin{equation}\label{eq:Hn_minimal2}
(u_n^{i},u_n^{i+1})=\frac{J_{1}}{4}\quad\hbox{and}\quad(u_n^{i},u_n^{i+2})=\frac{J_{1}^{2}}{8}-1
\end{equation}
for all $i=0,\dots, \left[1/\lan\right]-2$.
\end{proposition}

\begin{proof}
Let $H_n$ be defined as in \eqref{Hn}. Since $H_n \ge 0$, by \eqref{enhn} we deduce $E_n(u)\ge -\left(1+\frac{J_{1}^{2}}{8}\right)(1-c_{n}\lan)$ for all $u \in L^\infty(I,S^{1})$. Now, fix $\varphi \in [-\frac \pi2, \frac \pi2]$ so that $\cos(\varphi)=\frac{J_{1}}{4}$. Then, we construct $u_n \in C_n(I,S^{1})$ by setting, for all $i=0,\dots, \left[1/\lan\right]$,
$$
u^i_n=(\cos(\varphi i), \sin(\varphi i))\,.
$$
By the prosthaphaeresis formulas we get
$$
u^i_n+u^{i+2}_n= 2 \cos(\varphi) u^{i+1}_n= \frac{J_{1}}{2}u^{i+1}_n
$$
for all $i=0,\dots, \left[1/\lan\right]-2$. This implies $H_n(u_n)=0$, thus $E_n(u_n)=-\left(1+\frac{J_{1}^{2}}{8}\right)(1-c_{n}\lan)$ and \eqref{minimo} follows.

Consider now a minimizer $u_n$ of $E_n$ over $L^\infty(I,S^{1})$. By definition of $E_{n}$, we have that $u_{n}\in C_n(I,S^{1})$. By \eqref{minimo} it must be $H_n(u_n)=0$, which in turn implies
\begin{equation}\label{eq:Hn_minimal1}
u^{i+1}_n=\frac{2}{J_{1}}(u^{i}_n+u^{i+2}_n)
\end{equation}
for all $i=0,\dots, \left[1/\lan\right]-2$. Since $u_n$ takes values on the unit sphere, by taking the modulus squared in \eqref{eq:Hn_minimal1} we further get that 
$$
1=\frac{4}{J_{1}^{2}}|u^{i}_n+u^{i+2}_n|^{2}=\frac{8}{J_{1}^{2}}(1+ (u^{i}_n,u^{i+2}_n))\,,
$$
so that 
\begin{equation*}
(u^{i}_n,u^{i+2}_n)=\frac{J_{1}^{2}}{8}-1.
\end{equation*}
By this and \eqref{eq:Hn_minimal1} we also get
\begin{equation*}
(u^{i}_n,u^{i+1}_n)=\frac{2}{J_{1}}(u^{i}_n,u^{i}_n+u^{i+2}_n)=\frac{2}{J_{1}}(1+(u^{i}_n,u^{i+2}_n))=\frac{J_{1}}{4}\,,
\end{equation*}
as required.
\end{proof}

\begin{remark}\label{rem:ferro}
Note that the case $J_{1}>4$ is trivial. In fact the ground states are all ferromagnetic, that is $u_{n}^{i}=\bar u$ for all $i=0,\dots,[1/\lan]$ and for some $\bar u\in S^{1}$. Indeed in this case, set $E_{n}^{(J_{1}=4)}$ the energy in \eqref{energy_{3}} for $J_{1}=4$, we have that, for all $u\in\bUn$ 
\begin{eqnarray}\label{enremark}
E_n(u)&=&-J_{1}\sum_{i=0}^{\left[1/\lan\right]-2}\lan(u^{i},u^{i+1})+\sum_{i=0}^{{\left[1/\lan\right]-2}}\lan(u^{i},u^{i+2})\\
&=&E_{n}^{(J_{1}=4)}(u)-(J_{1}-4)\sum_{i=0}^{\left[1/\lan\right]-2}\lan(u^{i},u^{i+1}).
\end{eqnarray}
By the previous proposition $E_{n}^{(J_{1}=4)}$ is minimized on uniform states, which trivially also holds true for the second term in the above sum. In particular the minimal value can be straightforwardly computed: 
\begin{equation*}
\min E_{n}(u)=-\left(J_{1}-1\right)(1-c_{n}\lan).
\end{equation*}
\end{remark}

%%%%%
%\frac{J_{1}^{2}}{8}=2
%%%%

\end{subsection}
\begin{subsection}{Zero order estimates}
The following theorem is the main result of this section. 

\begin{theorem}\label{zeroorder}
Let $E_{n}:L^\infty(I,S^{1})\to(-\infty,+\infty]$ be the functional in \eqref{energy_{1}}. Then $\Gamma\hbox{-}\lim_{n}E_{n}(u)$ with respect to the weak-$*$ convergence in $L^{\infty}(I)$ is given by 
\begin{equation}\label{Gamma-lim}
E(u):=\begin{cases}
\int_{I}f_{hom}(u(x))\ dx&\textit{if } |u|\leq 1,\\
+\infty&\textit{otherwise in } L^{\infty}(I,\R^{2}),
\end{cases}
\end{equation}
where the convex function $f_{hom}:B_{1}\to\R$ is given by the following asymptotic homogenization formula:
\begin{eqnarray}\label{fhom}
f_{hom}(z):=\lim_{\rho\to 0}\lim_{k\to\infty}\frac{1}{k}\inf_{u\in {\overline{\mathcal U}}(Q_{h})}\{-J_{1}\sum_{i=0}^{k-2}(u^{i},u^{i+1})+\sum_{i=0}^{k-2}(u^{i},u^{i+2}),\  <u>_{1,Q_{h}}\in{\overline B(z,\rho)}\}.
\end{eqnarray}
Furthermore
\begin{itemize}
\item[(i)] if $J_{1}\geq 4$ then $f_{hom}(z)=-(J_{1}-1)$,
\item[(ii)] if $0<J_{1}\leq 4$ then the following estimate hold:
\begin{equation}
\frac{(J_{1}-4)^{2}}{8}|z|^{2}\leq f_{hom}(z)+(1+\frac{J_{1}^{2}}{8})\leq \frac{(J_{1}-4)^{2}}{8}|z|.
\end{equation}
Moreover there exists $h:[0,1]\to \R$ convex and monotone non-decreasing such that $f_{hom}(z)=h(|z|)$.
\item[(iii)] if $0<J_{1}\leq 4$ we have that $\min E(u)=E(0)=-(1-\frac{J_{1}^{2}}{8})$. 
\end{itemize}
\end{theorem}
\begin{proof}
The formula in \eqref{fhom} follows applying Theorem \ref{32-th:homog-Li} in the special case 
\begin{eqnarray}
f^{\xi}(u,v)=\begin{cases}
-\frac{J_{1}}{2}(u,v) &\text {if } |\xi|=1,\\
\frac{1}{2}(u,v) &\text {if } |\xi|=2,\\
0 &\text {otherwise}\\
\end{cases}
\end{eqnarray}
and $K=S^{1}$. To prove (i) we notice that, as observed in Remark \ref{rem:ferro}, $E_{n}$ is minimized by constant $S^{1}$-valued functions and its minimum is $-(J_{1}-1)(1-c_{n}\lan)$. Since $E_{n}$ $\Gamma$-converges to $E$ given by \eqref{Gamma-lim} we have that 
\begin{eqnarray}
f_{hom}(z)\geq -(J_{1}-1), \quad \forall z\in B^{1},\\
f_{hom}(z)=-(J_{1}-1), \quad \forall z\in S^{1}.
\end{eqnarray}
By the convexity of $f_{hom}$, (i) follows.\\

We divide the proof of (ii) into the lower bound and the upper bound estimates.\\

\noindent {\bf Lower bound:} let $u_{n}\in C_{n}(I,S^{1})$ be such that $u_{n}\wtos u$ in $L^{\infty}(I,\R^{2})$, by \eqref{enhn} it is left to prove that 
\begin{equation}\label{stima1}
\liminf_{n}H_{n}(u_{n})\geq \frac{(J_{1}-4)^{2}}{8}\int_{I}|u(x)|^{2}\ dx.
\end{equation}
We define the functions $w_{n}$ to be piece-wise constant on the cells of the lattice and such that 
\begin{equation}
w_{n}^{i}=\begin{cases}\frac{u_{n}^{i}+u_{n}^{i+2}}{2}&\textit{if } i=0,\dots,[1/\lan]-2,\\
0&\textit{if }i=[1/\lan]-1.
\end{cases}
\end{equation}
Let us show that $w_{n}\wtos u$ in $L^{\infty}(I,\R^{2})$. Since $\sup_{n}\|w_{n}\|_{\infty}\leq 1$ and $u_{n}\wtos u$ in $L^{\infty}(I,\R^{2})$, it suffices to show that, for all $(a,b)\subset\subset I$ it holds 
\begin{equation}
\int_{a}^{b}(u_{n}(x)-w_{n}(x))\ dx\to 0.
\end{equation}
The above limit follows on observing that
\begin{eqnarray*}
\left|\int_{a}^{b}(u_{n}(x)-w_{n}(x))\ dx\right|&\leq& \left|\sum_{i=[a/\lan]}^{[b/\lan]}\lan(u^{i}_{n}-w^{i}_{n})\right|+o(1)\\&=&\frac{1}{2}\left|\sum_{i=[a/\lan]}^{[b/\lan]}\lan(u^{i}_{n}-u^{i+2}_{n})\right|+o(1)
\leq4\|u_{n}\|_{\infty}\lan+o(1)\to 0.
\end{eqnarray*}
We also need to define the functions $\hat u_{n}$ piece-wise constant on the cell of the lattice and such that $\hat u_{n}^{i}:=u_{n}^{i+1}$. An analogous computation as the one above shows that $\hat u_{n}\wtos u$ in $L^{\infty}(I,\R^{2})$. We now may write that 
\begin{equation}
H_{n}(u_{n})=2\sum_{i=0}^{[1/\lan]-2}\lan\left|\frac{u_{n}^{i}+u_{n}^{i+2}}{2}-\frac{J_{1}}{4}u_{n}^{i+1}\right|^{2}\geq 2\int_{I}|w_{n}(x)-\tfrac{J_{1}}{4}\hat u_{n}(x)|^{2}\ dx+o(1).
\end{equation} 
By the weak lower semicontinuity of the $L^{2}$ norm we deduce \eqref{stima1}.\\

\noindent {\bf Upper bound:} we first prove that $f_{hom}(0)=-(1-\frac{J_{1}^{2}}{8})$. Using the already proved lower bound in $(ii)$ it is left to show that $f_{\hom}(0)\leq-(1-\frac{J_{1}^{2}}{8})$. To this end we construct the sequence of piecewise-constant functions $u_{n}$ on the cells of the lattice such that $u_{n}^{i}=(\cos\varphi i,\sin\varphi i)$. It holds that $u_{n}\wtos 0$ and moreover, as shown in Proposition \ref{minimi} $E_{n}(u_{n})=(1-c_{n}\lan)(-1-\frac{J_{1}^{2}}{8})$. As a result
\begin{equation}
f_{hom}(0)=\int_{I}f_{hom}(0)\ dx\leq \liminf_{n} E_{n}(u_{n})=-1-\frac{J_{1}^{2}}{8}.
\end{equation}
We now prove the upper bound for $z\in S^{1}$. Let us consider a constant sequence $u_{n}=z$. Using formula \eqref{energy_{3}} and \eqref{enhn}, we have that 
\begin{equation*}
E_{n}(u_{n})=(1-c_{n}\lan)(-1-\frac{J_{1}^{2}}{8}+\frac{(J_{1}-4)^{2}}{8}). 
\end{equation*}
Arguing as before, it follows that, for all $z\in S^{1}$, $f_{hom}(z)+(1+\frac{J_{1}^{2}}{8})\leq \frac{(J_{1}-4)^{2}}{8}$.

%%%%%%
Now for all $z\in B^{1}$ the upper bound follows by the convexity of $f_{hom}$. \\

Finally, by the definition of $f_{hom}$ it follows that, for all $z\in B^{1}$ $f_{hom}(Rz)=f_{hom}(z)$ for all $R\in SO(2)$. As a consequence of this and \cite[Corollary 12.3.1 and Example below]{Roc} we also get that $f_{hom}(z)=h(|z|)$ for some $h:[0,1]\to \R$ convex and monotone non-decreasing. Eventually (iii) follows by (ii). \end{proof}
\begin{remark}
We notice that $0$ is the unique minimizer of $f_{hom}$, in all the cases when the $\Gamma$-limit is non trivial, that is for  $0<J_{1}<4$.
\end{remark}

\end{subsection}

\section{Renormalization of the energy close to the ferromagnetic state and chirality transitions}\label{sec:grad}

In this section, motivated by the study of spin systems close to the helimagnet/ferromagnet transition point, we let the ferromagnetic  interaction parameter $J_{1}$ be scale dependent and approach the transition value $4$ from below. Namely we set $J_{1}=J_{1,n}=4(1-\delta_{n})$ for some $\den>0$, $\delta_{n}\to 0$. We then perform a renormalization of the energy $E_{n}$ and introduce a new functional whose asymptotic behavior will better describe the ground states of the system. More precisely we define $E_{n}^{hf}:L^{\infty}(I,\R^{2})\to(-\infty,+\infty]$ and $H_{n}^{hf}:L^{\infty}(I,\R^{2})\to[0,+\infty]$as:
\begin{equation}\label{Endelta}
E^{hf}_{n}(u):=\begin{cases}
\dps-4(1-\den)\hspace{-.3cm}\sum_{i=0}^{\left[1/\lan\right]-2}\lan(u^{i},u^{i+1})+\hspace{-.3cm}\sum_{i=0}^{{\left[1/\lan\right]-2}}\lan(u^{i},u^{i+2})&\textit{if }u\in C_n(I,S^{1})\\
+\infty&\textit{otherwise.}
\end{cases}
\end{equation}
\begin{equation}\label{Hndelta}
H^{hf}_{n}(u):=\begin{cases}
\frac{1}{2}\sum_{i}\lan\left|u^{i+2}-2(1-\den)u^{i+1}+u^{i}\right|^{2}&\textit{if }u\in C_n(I,S^{1})\\
+\infty&\textit{otherwise.}
\end{cases}
\end{equation}
Note that by Theorem \ref{zeroorder} it holds
\begin{equation}\label{Ehf-to-Hhf}
H^{hf}_{n}(u)=E^{hf}_{n}(u)-\min E_{n}^{hf}=E^{hf}_{n}(u)+(3-4\den+2\den^{2})(1-c_{n}\lan)
\end{equation}
\begin{proposition}\label{trivial-limit}
Let $E_{n}^{hf}:L^\infty(I,S^{1})\to(-\infty,+\infty]$ be the functional in \eqref{Endelta}. Then $\Gamma\hbox{-}\lim_{n}E_{n}(u)$ with respect to the weak-$*$ convergence in $L^{\infty}$ is given by 
\begin{equation*}
E(u):=\begin{cases}
-3&\textit{if } |u|\leq 1,\\
+\infty&\textit{otherwise in } L^{\infty}(I,\R^{2}).
\end{cases}
\end{equation*}
\end{proposition}
\begin{proof}
Observing that for all $u\in C_{n}(I,S^{1})$ it holds that
\begin{equation}
|E^{hf}_{n}(u)-E^{{J_{1}=4}}_{n}|\leq 4\den,
\end{equation}
the result immediately follows by Theorem \ref{zeroorder}.\end{proof}\\

\noindent In what follows we will define a convenient order parameter such that the $\Gamma$-limit of a scaled version $H_{n}^{hf}$ is given by a functional penalizing the helimagnetic transition around the ferromagnetic state.\\

We first introduce the order parameter. Given $u_{n}\in C_{n}(I,S^{1})$, for $i=0,1,\dots,[1/\lan]-1$ we associate to each $u_{n}^{i},u_{n}^{i+1}$ the corresponding oriented central angle $\theta_{n}^{i}\in[-\pi,\pi)$ given by 
\begin{equation}
\theta_{n}^{i}:=\chi[u_{n}^{i},u_{n}^{i+1}]\arccos((u_{n}^{i},u_{n}^{i+1})).
\end{equation}
We now set
\begin{equation}\label{wn}
w_{n}^{i}=\sin\left(\frac{\theta_{n}^{i}}{2}\right)\,.
\end{equation}
We eventually define the order parameter of our problem as
\begin{equation}\label{zn}
z_{n}^{i}={\sqrt{\frac{2}\den}}w_{n}^{i}.
\end{equation}
Note that, the above procedure defines $T_{n}:u_{n}\mapsto z_{n}$ which associate to any given $u_{n}\in C_{n}(I,S^{1})$ a piecewise-constant function $z_{n}\in \tilde C_{n}(I,\R)$ where 
\begin{eqnarray*}
\tilde C_{n}(I,\R):=\{z:[0,1)\to\R:z(x)=z^{i}_{n}, \hbox{ if } x \in \lan\{i+[0,1)\},\,i=0,1,\dots,[1/\lan]-1\}
\end{eqnarray*}
with  $z_{n}$ as in \eqref{zn}. 
We observe that if $z_{n}=T_{n}(u_{n})=T_{n}(v_{n})$ then $u_{n}$ and $v_{n}$ differ by a constant rotation (possibly depending on $n$) so that $H^{hf}_{n}(u_{n})=H^{hf}_{n}(v_{n})$. Therefore, with a slight abuse of notation, we now regard $H^{hf}_{n}$ as a functional defined on $z\in L^{1}(I,\R)$ by 
\begin{equation}\label{hnhf}
H^{hf}_{n}(z)=\begin{cases}
H^{hf}_{n}(u),&\textit{if } z\in\tilde C_{n}(I,\R)\\
+\infty   &\textit{ otherwise.}
\end{cases}
\end{equation}
Note that in the definition above, $u$ is any function such that $T_{n}u=z$. As a consequence, it will be natural to state the $\Gamma$-convergence theorem considering the convergence with respect to the order parameter $z$. 
\begin{theorem}\label{maintheorem}
Let $H_{n}^{hf}:L^1(I,\R)\to[0,+\infty]$ be the functional in \eqref{hnhf}. Assume that there exists
$l:=\lim_{n}\lan/(2\den)^{1/2}$. Then $H^{hf}(z):=\Gamma\hbox{-}\lim_{n}H_{n}^{hf}(z)/(\sqrt{2}\lan\den^{3/2})$ with respect to the $L^{1}(I)$ convergence is given by one of the following formulas:
\begin{itemize}
\item[(i)] if $l=0$ 
\begin{equation}\label{Gamma-lim_linfty}
H^{hf}(z):=\begin{cases}
\frac{8}{3}\#(S(z))&\textit {if } z\in BV(I,\{\pm 1\}), \\
+\infty&\textit{otherwise.}
\end{cases}
\end{equation}
\item [(ii)] if $l\in(0,+\infty)$ 
\begin{equation}
H^{hf}(z):=\begin{cases}
\frac1{l}\int_{I}(z^{2}(x)-1)^{2}\ dx+{l}\int_{I}(z'(x))^{2}\ dx &\textit {if } z\in W_{|per|}^{1,2}(I), \\
+\infty&\textit{otherwise,}
\end{cases}
\end{equation}

where we have set $W^{1,2}_{|per|}(I):=\{z\in W^{1,2}(I),\ s.t.\ |z(0)|=|z(1)|\}$.\\
\item[(iii)] if $l=+\infty$ 
\begin{equation}\label{Gamma-lim_l=0}
H^{hf}(z):=\begin{cases}
0&\textit {if } z=const, \\
+\infty&\textit{otherwise.}
\end{cases}
\end{equation}
\end{itemize}
\end{theorem}

In the following proposition we consider an equi-bounded sequence of spins and obtain a first bound on the scalar product between neighbors. 

\begin{proposition}\label{apriori}
Let $\mun\to 0$ and let $u_{n}$ be such that 
\begin{equation}\label{equibound}
\sup_{n}H_{n}(u_{n})\leq C\lan\mun,
\end{equation}
then for all $i$ we have that
\begin{equation}\label{stimaapriori-punt}
|\frac{J_{1}}{4}-(u_{n}^{i},u_{n}^{i+1})|\leq \sqrt{C}\left(\frac2{J_{1}}+\frac{1}{2}\right)\mun^{1/2}
\end{equation}

\end{proposition}
\begin{proof}
Since for all $i$ we have that 
\begin{equation*}
\left|u^{i+2}-\frac{J_{1}}{2}u^{i+1}+u^{i}\right|^{2}\geq \left(\left|u^{i}-\frac{J_{1}}{2}u^{i+1}\right|-1\right)^{2},
\end{equation*}
by \eqref{equibound} and the definition of $H_{n}$ we have that
\begin{equation*}
\sum_{i}\lan\left(\left|u^{i}-\frac{J_{1}}{2}u^{i+1}\right|-1\right)^{2}\leq C\lan\mun
\end{equation*}
which implies that, for all $i$,
\begin{equation*}
\left(\left|u^{i}-\frac{J_{1}}{2}u^{i+1}\right|-1\right)^{2}\leq C\mun.
\end{equation*}
As a result we have that
\begin{equation*}
\left(\left|u^{i}-\frac{J_{1}}{2}u^{i+1}\right|^{2}-1\right)^{2}\leq C\left(2+\frac{J_{1}}{2}\right)^{2}\mun.
\end{equation*}
By an explicit computation we finally get \eqref{stimaapriori-punt}. 
\end{proof}

\begin{proof}[Proof of Theorem \ref{maintheorem}] 
We prove the theorem only in cases $(i)$ and $(ii)$, since the proof of $(iii)$ involves only minor changes of the arguments we need in the other two cases.\\

\noindent Let us consider a sequence $z_n\in \tilde C_{n}(I,\R)$ such that $\sup_{n}\frac{H_{n}^{hf}(z_{n})}{\lan\den^{3/2}}\leq C <+\infty$. Equivalently there is a sequence $u_{n}\in C_{n}(I,S^{1})$ satisfying $\sup_{n}\frac{H_{n}^{hf}(u_{n})}{\lan\den^{3/2}}\leq C <+\infty$. We claim that 

\begin{equation}\label{liminf}
\frac{H^{hf}_{n}(u_{n})}{\sqrt{2}\lan\den^{3/2}}\geq
\frac{\sqrt{2}\den^{1/2}}{\lan}
\sum_{i=0}^{[1/\lan]-2}\lan\left((z_{n}^{i})^{2}-1\right)^{2}+\frac{\lan}{\sqrt{2}\den^{1/2}}(1-\gamma_{n})\hspace{-.3cm}\sum_{i=0}^{[1/\lan]-2}\lan\left(\frac{z_{n}^{i+1}-z_{n}^{i}}{\lan}\right)^{2}
\end{equation}
for some $\gamma_{n}\to 0$. 
Associating to each $u_{n}^{i}$ the angles $\theta^{i}_{n}$ and the functions $w^{i}_{n}$ introduced in \eqref{wn}, by means of the trigonometric identity $1-\cos(2x)=2\sin^{2}(x)$ we can write that
\begin{eqnarray*}
1-(u_{n}^{i},u_{n}^{i+1})&=&1-\cos(\theta_{n}^{i})=2\sin^{2}\left(\frac{\theta_{n}^{i}}{2}\right)=2(w_{n}^{i})^{2}\\
1-(u_{n}^{i},u_{n}^{i+2})&=&1-\cos(\theta_{n}^{i+1}+\theta_{n}^{i})\,.
\end{eqnarray*}
By Lemma \ref{apriori} with $\mun=\den^{\frac32}$ there exists a constant $C'$ such that
\begin{equation}\label{conv-zero}
1-(u_{n}^{i},u_{n}^{i+1})\leq C'\den ^{\frac34},
\end{equation}
so that in particular $\theta_{n}^{i}\to 0$.

Introducing the function $w_{n}$ and the angles $\theta_{n}$, by\eqref{cn} and \eqref{Ehf-to-Hhf} we may rewrite $H^{hf}_{n}(u_{n})$ as follows
\begin{eqnarray}\label{Hn_w}
H^{hf}_{n}(u_{n})&=&4(1-\den)\sum_{i=0}^{[1/\lan]-2}\lan(1-(u_{n}^{i},u_{n}^{i+1}))-\sum_{i=0}^{[1/\lan]-2}\lan(1-(u_{n}^{i},u_{n}^{i+2}))\nonumber \\&&+2\den^{2}(1-c_{n}\lan)\nonumber \\
&=&8(1-\den)\sum_{i=0}^{[1/\lan]-2}\lan(w_{n}^{i})^{2}-\sum_{i=0}^{[1/\lan]-2}\lan(1-\cos(\theta_{n}^{i+1}+\theta_{n}^{i}))\nonumber \\&&+2\den^{2}(1-c_{n}\lan)\nonumber .
\end{eqnarray}
We further point out the following identities:
\begin{eqnarray*}
&\vphantom{\displaystyle\sum_{i=0}^{[1/\lan]-2}}
4(w_{n}^{i})^{2}-\sin^2(\theta_{n}^{i})=4(w_{n}^{i})^{4},\\
&\nonumber \displaystyle
2\hspace{-.35cm}\sum_{i=0}^{[1/\lan]-2}\hspace{-.35cm}\lan\sin^{2}(\theta_{n}^{i})=\hspace{-.35cm}\sum_{i=0}^{[1/\lan]-2}\hspace{-.35cm}\lan(\sin^{2}(\theta_{n}^{i})+\sin^{2}(\theta_{n}^{i+1})).
\end{eqnarray*}
The first one comes from the trigonometric identity $4\sin^{2}(x)-\sin^2(2x)=4\sin^4(x)$ while the second follows from the boundary condition \eqref{boundary}. Moreover the following limit holds true
\begin{equation}\label{limit}
\lim_{\substack{(x,y)\to(0,0)\\x\not=y}}\frac{\sin^{2}(x)+\sin^{2}(y)-(1-\cos(x+y))}{(\sin(x/2)-\sin(y/2))^{2}}=2
\end{equation}
upon observing that 
\begin{equation*}
\sin^{2}(x)+\sin^{2}(y)-(1-\cos(x+y))=(\sin(x)-\sin(y))^{2}-(1-\cos(x-y)).
\end{equation*}
We can therefore continue estimate $H^{hf}_{n}(u_{n})$ as
\begin{eqnarray}
H^{hf}_{n}(u_{n})&=&\sum_{i=0}^{[1/\lan]-2}\lan(8(w_{n}^{i})^{2}-2\sin^{2}(\theta_{n}^{i}))-8\den\sum_{i=0}^{[1/\lan]-2}\lan(w_{n}^{i})^{2}\nonumber \\&&+ 2\sum_{i=0}^{[1/\lan]-2}\lan\sin^{2}(\theta_{n}^{i})-\sum_{i=0}^{[1/\lan]-2}\lan(1-\cos(\theta_{n}^{i+1}+\theta_{n}^{i}))\nonumber \\&&+2\den^{2}(1-c_{n}\lan)\nonumber \\
&=&8\sum_{i=0}^{[1/\lan]-2}\lan\left((w_{n}^{i})^{4}-\den(w_{n}^{i})^{2}+\frac{\den^{2}}{4}\right)
+ \sum_{i=0}^{[1/\lan]-2}\lan \left(\sin^{2}(\theta_{n}^{i+1})+\sin^{2}(\theta_{n}^{i})\right)\nonumber \\ \nonumber&&-\sum_{i=0}^{[1/\lan]-2}\lan(1-\cos(\theta_{n}^{i+1}+\theta_{n}^{i})) \nonumber\\
&\geq&8\sum_{i=0}^{[1/\lan]-2}\lan\left((w_{n}^{i})^{2}-\frac{\den}{2}\right)^{2}\nonumber + 2(1-\gamma_{n})\sum_{i=0}^{[1/\lan]-2}\lan (w_{n}^{i+1}-w_{n}^{i})^{2}, \nonumber
%\\
%&&+ \sum_{i=0}^{[1/\lan]-2}\lan \left(\sin^{2}([\theta_{n}^{i+2}-\theta_{n}^{i+1}])+\sin^{2}([\theta_{n}^{i+1}-\theta_{n}^{i}])-(1-\cos([\theta_{n}^{i+2}-\theta_{n}^{i+1}]+[\theta_{n}^{i+1}-\theta_{n}^{i}]))\right). \nonumber
\end{eqnarray}
for some $\gamma_{n}\to 0$. The last inequality is a consequence of \eqref{limit} once we recall that, by \eqref{conv-zero}, $\theta_{n}\to 0$ uniformly. In terms of $z_{n}$ the inequality in \eqref{Hn_w} becomes:
\begin{eqnarray}
H^{hf}_{n}(u_{n})&\geq&2\den^{2}\sum_{i=0}^{[1/\lan]-2}\lan\left((z_{n}^{i})^{2}-1\right)^{2} + (1-\gamma_{n})\den\sum_{i=0}^{[1/\lan]-2}\lan (z_{n}^{i+1}-z_{n}^{i})^{2}. \nonumber
\end{eqnarray}
The claim \eqref{liminf} is proved on dividing by $\sqrt{2}\lan\den^{3/2}$. \\

The claim implies the liminf inequality both in case $(i)$ and $(ii)$. In case $(i)$ it is obtained applying Theorem \ref{braides-yip} and Remark \ref{rem:W}. For what concerns $(ii)$, it suffices to observe that piecewise affine interpolations of the sequence $z_{n}$ associated to an equibounded $u_{n}$ are, in this case, weakly compact in $W_{|per|}^{1,2}(I)$ so that the lower bound follows by standard lower semicontinuity. \\

\medskip

\noindent In order to prove the limsup inequality we separately discuss cases $(i)$ and $(ii)$. \\

\noindent Case$(i)$. By the locality of the construction it suffices to exhibit a recovery sequence for $z=-\chi_{(0,1/2]}+\chi_{(1/2,1)}$.
As it is well known, (see for example \cite{Mod-Mor}, \cite{modica}) $z_{min}(t)=\tanh(t)$ is the solution of the following problem
\begin{equation}
\min\left\{\int_{-\infty}^{+\infty}((z'(t))^{2}+(z(t)^{2}-1)^{2})\ dt, z\in W^{1,2}(\R),\ z(\pm\infty)=\pm 1\right\}=:m
\end{equation}
and, by a direct computation, the above minimum is $m=\frac{8}{3}$. For all $\e>0$ there exists $R_{\epsilon}>0$ such that 
\begin{eqnarray}\label{prop-recovery}
\max\{\sup_{t\in(-\infty,-R_{\e})}|z_{min}(t)+1|,\sup_{t\in(R_{\e},+\infty)}|z_{min}(t)-1|\}\leq \e\\
\int_{-R_{\e}}^{+R_{\e}}((z_{min}'(t))^{2}+(z_{min}(t)^{2}-1)^{2})\leq m+\e\nonumber
\end{eqnarray}
Let us define $z_{\e}:\R\to\R$ as the odd $C^{1}$ function such that 
\begin{equation*}
z_{\e}(t):=\begin{cases}
z_{min}(t)&\textit{if }t\in[0,R_{\e}],\\
p_{\e}(t)&\textit{if }t\in(R_{\e},R_{\e}+\e),\\
1&\textit{if }t\in(R_{\e}+\e,+\infty),
\end{cases}
\end{equation*}
where $p_{\e}$ is a suitable third order interpolating polynomial that we may choose such that $\|p_{\e}'\|_{\infty}\leq 2$. 
\begin{figure}
\begin{center}
\includegraphics[scale=.65 ]{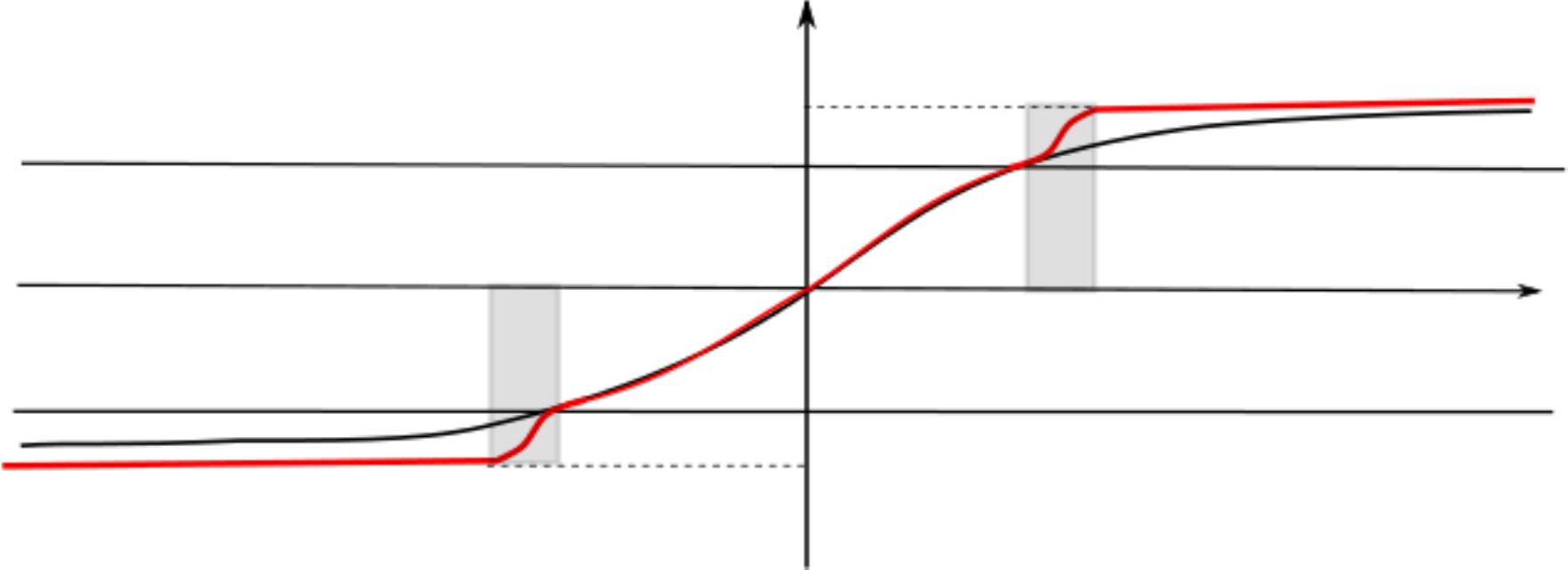}
\caption{In black the function $\tanh(t)$. In red the function $z_{\e}(t)$ in the limsup construction}\label{fig-limsup}
\end{center}
\begin{picture}(0,0)
\begin{footnotesize}
\put(270,91){$R_{\e}$}\put(287,91){$R_{\e}+\e$}
\put(223,140){$1$}\put(212,120){$1-\e$}
%\put(-63,159){$Q_3$}
%\put(14,167){$Q_4$}
%\put(-13,189){$\partial K$}
\end{footnotesize}
\end{picture}
\end{figure}
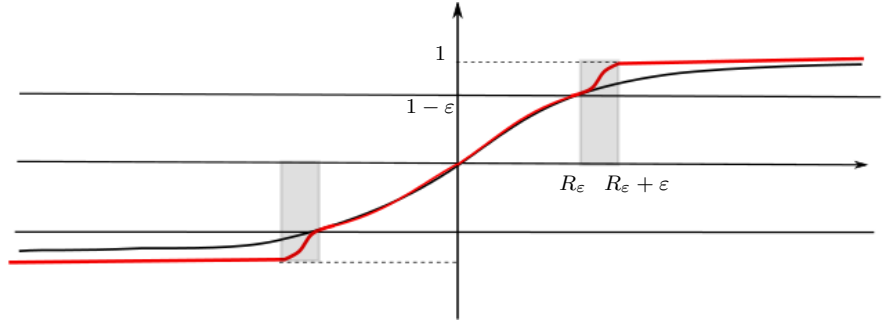
Note that, by the definition of $z_{\e}$ and by the properties \eqref{prop-recovery} above we have that there exists $C>0$ such that
\begin{equation}
\int_{-\infty}^{+\infty}((z_{\e}'(t))^{2}+(z_{\e}(t)^{2}-1)^{2})\ dt\leq m+C\e.
\end{equation}
Let $z_{n,\e}\in \tilde C_{n}(I,\R)$ be defined as follows
\begin{equation}
z_{n,\e}^{i}=z_{\e}\left(\frac{\sqrt{2\den}}{\lan}(\lan i-\frac{1}{2})\right).
\end{equation}
We have that $z_{n,\e}\to z$ in $L^{1}(I)$ as $n\to +\infty$. If we set
\begin{eqnarray*}
i_{+}=\left[\frac{1}{\lan}\left(\frac{1}{2}+\frac{\lan}{\sqrt{2\den}}(R_{\e}+\e)\right)\right]+1\quad\hbox{and}\quad
i_{-}=\left[\frac{1}{\lan}\left(\frac{1}{2}-\frac{\lan}{\sqrt{2\den}}(R_{\e}+\e)\right)\right]\,,
\end{eqnarray*}
then $|z^i_{n,\e}|=1$ for all $i\geq i_{+}$ or $i\leq i_{-}$. We now put $w_{n,\e}^{i}=\sqrt{\frac{\den}{2}}z_{n,\e}^{i}$, so that in particular for all $i$ $|w_{n,\e}^{i}|\leq \sqrt{\frac{\den}{2}}$. We can therefore define the angles 
\begin{equation*}
\varphi_{n,\e}^{i}=2\sum_{j=0}^{i}\arcsin(w_{n,\e}^{j}).
\end{equation*}
Let us observe that $\sgn(\varphi_{n,\e}^{i+1}-\varphi_{n,\e}^{i})=\sgn(w_{n}^{i})$ and that $\varphi_{n,\e}^{1}-\varphi_{n,\e}^{0}=\varphi_{n,\e}^{[1/\lan]}-\varphi_{n,\e}^{[1/\lan]-1}$. As a consequence, upon defining $u_{n,\e}^{i}=(\cos(\varphi_{n,\e})^{i},\sin(\varphi_{n,\e}^{i}))$, we have that $u_{n}\in\bUn$ and that $T_{n}(u_{n,\e})=z_{n,\e}$. Using again the limit \eqref{limit} and repeating the computation as in the proof of the lower bound, we obtain the existence of a sequence $\eta_{n}\to 0$ such that
\begin{eqnarray}\label{limsup1}
\frac{H^{hf}_{n}(u_{n,\e})}{\sqrt{2}\lan\den^{3/2}}\leq
\frac{\sqrt{2}\den^{1/2}}{\lan}
\sum_{i=0}^{[1/\lan]-2}\lan\left((z_{n,\e}^{i})^{2}-1\right)^{2}+\frac{\lan}{\sqrt{2}\den^{1/2}}(1+\eta_{n})\hspace{-.3cm}\sum_{i=0}^{[1/\lan]-2}\lan\left(\frac{z_{n,\e}^{i+1}-z_{n,\e}^{i}}{\lan}\right)^{2}.
\end{eqnarray}

Define now the piecewise constant functions $z^1_{\e,n}(s)$ via
$$
z^1_{\e,n}(s):=\begin{cases}
               \left(\frac{z_{n,\e}^{i+1}-z_{n,\e}^{i}}{\sqrt{2}\den^{1/2}}\right)&\hbox{ if }s\in \left[\frac{\sqrt{2}\den^{1/2}}{\lan}(\lan i-\frac12), \frac{\sqrt{2}\den^{1/2}}{\lan}(\lan (i+1)-\frac12)\right)\,,i=0,\dots,[1/\lan]-2\\
               0 &\hbox{ otherwise.}
               \end{cases}
$$
Notice that by constuction $z^1_{\e,n}(s)=0$ when $|s|\ge R_\e+\e+2\lan$.\\ Since each of the intervals $\left[\frac{\sqrt{2}\den^{1/2}}{\lan}(\lan i-\frac12), \frac{\sqrt{2}\den^{1/2}}{\lan}(\lan (i+1)-\frac12)\right)$ has length $\sqrt{2}\den^{1/2} \to 0$, and since $z_{\e}'$ is uniformly continuous in $\R$, we get that $|z^1_{\e,n}(s)-z_{\e}'(s)|\to 0$ uniformly with respect to $s \in \R$. Being $z^1_{\e,n}(s)=0$ outside a compact set independent of $n$, this implies
$$
\lim_{n\to +\infty}\int_{-\infty}^{+\infty}(z^1_{\e,n}(s))^2\,ds=\int_{-\infty}^{+\infty}(z_{\e}'(s))^2\,ds\,.
$$
On the other hand, by a direct computation, we get that
\begin{eqnarray*}
&\displaystyle
\frac{\lan}{\sqrt{2}\den^{1/2}}(1+\eta_{n})\hspace{-.3cm}\sum_{i=0}^{[1/\lan]-2}\lan\left(\frac{z_{n,\e}^{i+1}-z_{n,\e}^{i}}{\lan}\right)^{2}=(1+\eta_{n})\hspace{-.3cm}\sum_{i=0}^{[1/\lan]-2}\sqrt{2}\den^{1/2}\left(\frac{z_{n,\e}^{i+1}-z_{n,\e}^{i}}{\sqrt{2}\den^{1/2}}\right)^{2}\\
&\displaystyle
\le (1+\eta_{n})\int_{-\frac{\sqrt{2}\den^{1/2}}{2\lan}}^{\frac{\sqrt{2}\den^{1/2}}{2\lan}}(z^1_{\e,n}(s))^2\,ds\le(1+\eta_{n})\int_{-\infty}^{+\infty}(z^1_{\e,n}(s))^2\,ds
\end{eqnarray*}
so that
\begin{equation}\label{limsup2}
\limsup_{n\to +\infty}\frac{\lan}{\sqrt{2}\den^{1/2}}(1+\eta_{n})\hspace{-.3cm}\sum_{i=0}^{[1/\lan]-2}\lan\left(\frac{z_{n,\e}^{i+1}-z_{n,\e}^{i}}{\lan}\right)^{2}\le \int_{-\infty}^{+\infty}(z_{\e}'(s))^2\,ds\,.
\end{equation}

To estimate the other term we proceed in a similar way. We define the piecewise constant functions $\hat z_{\e,n}(s)$ via
$$
\hat z_{\e,n}(s):=\begin{cases}
               z_{n,\e}^{i}&\hbox{ if }s\in \left[\frac{\sqrt{2}\den^{1/2}}{\lan}(\lan i-\frac12), \frac{\sqrt{2}\den^{1/2}}{\lan}(\lan (i+1)-\frac12)\right)\,,i=0,\dots,[1/\lan]-2\\
               0 &\hbox{ otherwise.}
               \end{cases}
$$
Notice that by constuction $\hat z_{\e,n}(s)^2=1$ when $|s|\ge R_\e+\e+2\lan$.\\ Since each of the intervals $\left[\frac{\sqrt{2}\den^{1/2}}{\lan}(\lan i-\frac12), \frac{\sqrt{2}\den^{1/2}}{\lan}(\lan (i+1)-\frac12)\right)$ has length $\sqrt{2}\den^{1/2} \to 0$, and since $z_{\e}$ is uniformly continuous in $\R$, we get that $|\hat z_{\e,n}(s)-z_{\e}(s)|\to 0$ uniformly with respect to $s \in \R$. Being $\hat z_{\e,n}(s)^2=1$ outside a compact set independent of $n$, this implies
$$
\lim_{n\to +\infty}\int_{-\infty}^{+\infty}(\hat z_{\e,n}(s)^2-1)^2\,ds=\int_{-\infty}^{+\infty}(z_{\e}(s)^2-1)^2\,ds\,.
$$
On the other hand, since by construction
$$
\hat z_{\e,n}\left(\frac{\sqrt{2}\den^{1/2}}{\lan}(t-\frac12)\right)=z_{n,\e}^{i}\iff t\in [\lan i, \lan(i+1))
$$
via the change of variables $t-\frac 12= \frac{\lan}{\sqrt{2}\den^{1/2}}s$ we have
\begin{eqnarray*}
&\displaystyle
\int_{-\infty}^{+\infty}(\hat z_{\e,n}(s)^2-1)^2\,ds \ge \int_{-\frac{\sqrt{2}\den^{1/2}}{2\lan}}^{\frac{\sqrt{2}\den^{1/2}}{2\lan}}(\hat z_{\e,n}(s)^2-1)^2\,ds\\
&\displaystyle
= \frac{\sqrt{2}\den^{1/2}}{\lan}
\int_0^1\left(\hat z_{\e,n}\left(\frac{\sqrt{2}\den^{1/2}}{\lan}(t-\frac 12)\right)^2-1\right)^2\,dt\ge \frac{\sqrt{2}\den^{1/2}}{\lan}
\sum_{i=0}^{[1/\lan]-2}\lan\left((z_{n,\e}^{i})^{2}-1\right)^{2}
\end{eqnarray*}
so that
\begin{equation}\label{limsup3}
\limsup_{n\to +\infty}\frac{\sqrt{2}\den^{1/2}}{\lan}
\sum_{i=0}^{[1/\lan]-2}\lan\left((z_{n,\e}^{i})^{2}-1\right)^{2}\le\int_{-\infty}^{+\infty}(z_{\e}(s)^2-1)^2\,ds\,.
\end{equation}
Combining \eqref{limsup1}, \eqref{limsup2}, and \eqref{limsup3} we obtain
$$
\limsup_{n\to +\infty}\frac{H^{hf}_{n}(u_{n,\e})}{\sqrt{2}\lan\den^{3/2}}\leq \int_{-\infty}^{+\infty}((z_{\e}'(t))^{2}+(z_{\e}(t)^{2}-1)^{2})\ dt\leq m+C\e.
$$
This gives the required upper bound by arbitrariness of $\e$.

\medskip

\noindent Case $(ii)$. We argue by density. Let us consider $z\in W^{1,2}_{|per|}(I)\cap C^{\infty}(\overline I)$. We define
\begin{eqnarray}
z_{n}^{i}=\begin{cases}
z(\lan i)&\text{if } i=1,2,\dots,\left[\frac{1}{\lan}\right]-1,\\
z(1)&\text{if } i=\left[\frac{1}{\lan}\right].
\end{cases}
\end{eqnarray} 
Note that, by taking the piecewise affine interpolation of such a $z_{n}$ we have that
\begin{eqnarray*}
\lim_{n}\left(\frac{\sqrt{2}\den^{1/2}}{\lan}
\sum_{i=0}^{[1/\lan]-2}\lan\left((z_{n}^{i})^{2}-1\right)^{2}+\frac{\lan}{\sqrt{2}\den^{1/2}}(1-\gamma_{n})\hspace{-.3cm}\sum_{i=0}^{[1/\lan]-2}\lan\left(\frac{z_{n}^{i+1}-z_{n}^{i}}{\lan}\right)^{2}\right)=H^{hf}(z).
\end{eqnarray*}
To conclude, we construct $u_{n}$ as in the proof of $(i)$ and observe that \eqref{limsup1} still holds true.
\end{proof}

\bibliographystyle{plain}
\bibliography{bib-CS}

\end{document}